\documentclass{cc}
\usepackage{amsthm}

\contact{kolarmartin@fit.vutbr.cz}
\submitted{10 October 2021}
\title{On formally undecidable propositions in nondeterministic languages}
\titlehead{undecidable propositions in NL}
\author{Martin Kol\'a\v{r}}

\begin{abstract}
Any class of languages $\mathbf{L}$ accepted in time $\mathbf{T}$ has a counterpart $\mathbf{NL}$ accepted in nondeterministic time $\mathbf{NT}$. It follows from the definition of nondeterministic languages that $\mathbf{L} \subseteq \mathbf{NL}$. This work shows that every sufficiently powerful language in $\mathbf{L}$ contains a string corresponding to G\"{o}del's undecidable proposition, but this string is not contained in its nondeterministic counterpart. This inconsistency in the definition of nondeterministic languages shows that certain questions regarding nondeterministic time complexity equivalences are irrevocably ill-posed.
\end{abstract}


\begin{subject}
68Q15, 03D35
\end{subject}

\begin{document}


\section{Introduction}
\label{intro}

It is demonstrated here that $\exists p \in L : p \notin NL$ for some $L \in \mathbf{L}$ and its nondeterministic counterpart $NL \in \mathbf{NL}$. This result is accomplished constructively, without resorting to the law of excluded middle. This holds for any set $\mathbf{L}$, where

\begin{definition}

Any set of languages capable of representing $\omega$-consistent recursive class $c$ of \textit{formulae} can be referred to as $\mathbf{L}$. $L$ is any language in $\mathbf{L}$, $NL$ is any language in $\mathbf{NL}$.

\end{definition}

For example, $\mathbf{P}$ is an $\mathbf{L}$ set, and so is any superset of $\mathbf{P}$.

\clearpage   

\section{Undecidable Propositions in Deterministic and Nondeterministic Turing Machines}
\label{godel}

\cite{godel1934undecidable} proved that for every $\omega$-consistent recursive class $c$ of \textit{formulae} there correspond recursive \textit{class-signs} $r$, such that neither $\nu$ Gen $r$ nor neg($\nu$ Gen $r$) belongs to Flg($c$) (where $\nu$ is a \textit{free variable} of $r$).

If it can be shown that languages in $\mathbf{NL}$ do not contain recursive \textit{class-signs} $r$, such that neither $\nu$ Gen $r$ nor neg($\nu$ Gen $r$) belongs to Flg($c$), it follows they cannot be an $\omega$-consistent recursive class of formulae. On the other hand, every language in $\mathbf{L}$ capable of expressing and interpreting arbitrary $\omega$-consistent recursive class signs is itself an $\omega$-consistent recursive class of formulae. This is proven in the two following lemmas:

\begin{lemma}
Given the undecidable proposition $p$, $p \in L$
\end{lemma}

\begin{proof}
Demonstrating that languages in $\mathbf{L}$ are capable of interpreting class signs simply requires a program that interprets proofs of mathematical theorems, as set out by Principia Mathematica, ZFC, or $\omega$-consistency. Polynomial-time algorithms, for instance, have this power because proofs consist of axioms and primitive recursive substitutions only. Strings accepted by such an algorithm form a language $L$, which must contain the undecidable proposition $p$.

One such algorithm may thus be the Turing Machine which evaluates equality for given values of any Diophantine equation, which is composed of addition, multiplication, and exponentiation only. It is known by \cite{matijasevic1970enumerable} that the undecidable proposition is a Diophantine equation, so the undecidable proposition $p$ is in $L$.
\end{proof}

\begin{lemma}
Given $NL$, the nondeterministic counterpart of $L$, $p \notin NL$
\end{lemma}

\begin{proof}
As usual, \cite{cook2006p} defines $\mathbf{NL}$ through a binary checking relation $R(x,y) \subseteq \Sigma^* \times \Sigma^*_1$ which processes strings $x$ and $y$. Consider $R_\omega$ which checks $\omega$-consistency of a statement of mathematical logic $x$ with proof $y$. All problems $x$ for which a proof $y$ does not exist are \textbf{not} included in the domain. Therefore, languages in $\mathbf{NL}$ do not contain undecidable propositions.
\end{proof}

\begin{theorem}
The $\mathbf{L}$ vs $\mathbf{NL}$ problem is undecidable
\end{theorem}

\begin{proof}
As per lemmas 2.1 and 2.2, for some sufficiently expressive $L \in \mathbf{L}$ and its counterpart $NL \in \mathbf{NL}$, $\exists p \in L : p \notin NL$.

For any class of languages $\mathbf{L}$ of a given deterministic computational time complexity, there exists a nondeterministic counterpart $\mathbf{NL}$. It follows from the definition of $\mathbf{NL}$ that $\mathbf{L} \subseteq \mathbf{NL}$. 

Thus, without any further assumptions, the definition directly leads to the contradiction that for all sufficiently expressive $L \in \mathbf{L}$ and their counterparts $NL \in \mathbf{NL}$, $\exists p \in L : p \notin NL$ while $L \subseteq NL$. 
\end{proof}

This is a stronger result than $\neg(\mathbf{L} = \mathbf{NL})$. In fact, it shows that for all languages in $\mathbf{L}$ sufficiently expressive to decribe $R_\omega$, any question of the form \textit{Does $\mathbf{L} = \mathbf{NL}$?} is ill-posed, and therefore undecidable.

\section{Attempting to Circumvent Undecidability}
\label{more}

Is it possible to create a definition of deterministic computation as to explicitly exclude undecidable propositions?

It is well known (\cite{hofstadter1979godel}) that any complete system of mathematical logic will contain undecidable propositions, propositions whose validity within the language can be verified, but for which a proof cannot exist. Consider a specific undecidable proposition $p_1$ in language $L_0 \in \mathbf{L}$, which is circumvented by including $p_1$ or $\neg p_1$ in the axioms of $L_1 \in \mathbf{L}$.

However, $L_1$ will also be subject to G\"odel's first theorem, giving rise to a proposition $p_2$ on $L_1$. Repeating this procedure any number of times, be it in the language, meta-language, or meta*-language, does not circumvent undecidability. In other words, limiting $\Sigma^* = \Sigma^*_1$ will not incur loss of generality.



\noacknowledge


\bibliography{bibliography}

\end{document}